\documentclass[journal]{IEEEtran}
\usepackage[dvips]{graphicx}
\usepackage{amsmath, amsthm, amssymb, amsfonts, mathtools}
\usepackage{cases, multirow}
\usepackage{srcltx,color}
\usepackage{cite, url}
\usepackage{subcaption}

\usepackage{algorithm, algpseudocode}

\long\def\comment#1{}

\newtheorem{thm}{Theorem}

\newtheorem{corollary}{Corollary}

\newtheorem{remark}{Remark}

\def\figref#1{Fig.~\ref{#1}}
\def\be{\begin{equation} }
\def\ee{\end{equation} }

\title{\huge A Stochastic Geometry Based Approach to Modeling Interference Correlation in Cooperative Relay Networks}

\begin{document}

\author{
   \IEEEauthorblockN{Young Jin Chun, Simon L. Cotton, Mazen O. Hasna, and Ali Ghrayeb}
   \thanks{Y. J. Chun and S. L. Cotton are with the Wireless Communications Laboratory, ECIT Institute, Queens University Belfast, United Kingdom
   (Email: Y.Chun@qub.ac.uk, simon.cotton@qub.ac.uk).}
   \thanks{M. O. Hasna is with Department of Electrical Engineering, Qatar University, Doha, Qatar (Email: hasna@qu.edu.qa).} 
   \thanks{A. Ghrayeb is with Department of Electrical and Computer Engineering, Texas A\&M University at Qatar, Doha, Qatar 
    (Email: ali.ghrayeb@qatar.tamu.edu).}
   \thanks{This work was supported in part by the Engineering and Physical Sciences Research Council (EPSRC) under Grant References EP/H044191/1 and EP/L026074/1, and in part by the NPRP Grant 4-1119-2-427 from the Qatar National Research Fund (a member of Qatar Foundation). The statements made herein are solely the responsibility of the authors.}
 } 

\maketitle

{\it \bf Abstract}-
{ 
Future wireless networks are expected to be a convergence of many diverse network technologies and architectures, such as cellular networks, wireless local area networks, sensor networks, and device to device communications. Through cooperation between dissimilar wireless devices, this new combined network topology promises to unlock ever larger data rates and provide truly ubiquitous coverage for end users, as well as enabling higher spectral efficiency. However, it also increases the risk of co-channel interference and introduces the possibility of correlation in the aggregated interference that not only impacts the communication performance, but also makes the associated mathematical analysis much more complex. To address this problem and evaluate the communication performance of cooperative relay networks, we adopt a stochastic geometry based approach by assuming that the interfering nodes are randomly distributed according to a Poisson point process (PPP). We also use a random medium access protocol to counteract the effects of interference correlation. Using this approach, we derive novel closed-form expressions for the successful transmission probability and local delay of a relay network with correlated interference. As well as this, we find the optimal transmission probability $p$ that jointly maximizes the successful transmission probability and minimizes the local delay. Finally numerical results are provided to confirm that the proposed joint optimization strategy achieves a significant performance gain compared to a conventional scheme. 
}


\section{Introduction}

Cooperative relaying is an effective technique for improving the reliability and throughput of the traditional point-to-point communication. The approach was first proposed by Cover and El Gamal in \cite{Cover1979} and revisited in \cite{Sendonaris2003, Laneman2004}. In \cite{Laneman2004}, the authors proposed several relaying protocols, such as amplify-and-forward (AF), decode-and-forward (DF), and selection relaying, and evaluated the performance of these relaying protocols. 
In \cite{Laneman2003}, the authors combined relaying with space-time coding. 
More advanced relaying protocols, such as compress-and-forward and compute-and-forward, were introduced in \cite{Kramer2005, Kim2008, Aleksic2009, Sanderovich2009, Lim2011, Nazer2011}. In compress-and-forward relaying, the relay observe a vector-quantized signal and forward this information to the destination. 
In compute-and-forward (CF) relaying, the relays decode and forward linear equations of transmitted messages using the noisy linear combinations provided by the channel. Once the destination node receives enough linear combinations, it can successfully detect the desired messages \cite{Nazer2011}.
Most of the previous work on cooperative relaying has focused on orthogonal channel allocation or noise-limited fading environments, ignoring co-channel interference. However, this assumption is not realistic due to the high spectral reuse in practical wireless networks. The effect of co-channel interference on cooperative relaying has been studied in \cite{Krikidis2009, Zhong2010, Si2010} assuming fixed interfering node deployments. In \cite{Krikidis2009}, the authors considered a relay network with interference affecting only the relay. While AF and DF relaying in an interference network were investigated in \cite{Zhong2010} and \cite{Si2010}, respectively.

In a real wireless network, however, it is more practical to assume that the interference and node locations are random due to mobility. Moreover, analyzing a specific instance of the network with a fixed node deployment does not provide a general result. Instead, a statistical statement about the ensembles of all possible node deployments is much more beneficial to assist in understanding the network performance. In this context, stochastic geometry has recently gained much attention. This method models interference in the network by treating the locations of the interferer as points distributed according to a spatial point process \cite{Haenggi2013}. Such an approach captures the topological randomness in the network geometry, provides well-established mathematical tools, allows high analytical flexibility and achieves an accurate performance evaluation. A common assumption in most of the related works is that the interfering nodes are distributed according to a homogeneous Poisson point process (PPP). Two important properties of PPP are stationarity, \textit{i.e.}, its distribution is invariant under arbitrary translation, and the Slivnyak’s theorem, which means that conditioning on a certain point does not change the distribution of the process \cite{Haenggi2013}. Due to these two properties, the PPP model is analytically tractable and flexible. The probability generating functional (PGFL) of PPP is derived in closed-form, and the distribution of the inter-node distance is known \cite{Haenggi2005}. The Laplace transform of the interference in a PPP network as well as the probability density function of the aggregated interference were analyzed for Rayleigh fading channels in \cite{Mathar1995, Andrews2011}. The outage probability and average achievable rate of heterogeneous cellular networks (HetNets) were evaluated for PPP and Poisson cluster process (PCP) in \cite{Jo2012} and \cite{Chun2015}, respectively. Cooperative relaying with PPP distributed interfering nodes has been investigated in \cite{Kountouris2009, Ganti2009, Cho2011, Yu2010}. In \cite{Kountouris2009}, the authors derived the throughput scaling law for opportunistic relay selection, while in \cite{Ganti2009}, decentralized relay selection schemes based on the location information or the received signal strength were proposed. To ensure a certain quality of service (QoS) at the destination node, the authors of \cite{Cho2011} have defined a QoS region for random relay selection.

Although assuming a homogeneous PPP offers a convenient method to model a network with uniformly distributed interfering nodes, using the PPP alone is not enough to accurately capture the real aspects of practical wireless networks. An important example of this occurs in the statistical behavior of the aggregated interference when the non-desired signals are correlated in either the space or time domain. Interference usually originates from a set of transmitters sharing common randomness, causing correlation in the aggregated interference. 
The impact of interference correlation was not properly reflected in the stochastic geometry based modeling process until very recently \cite{Ganti2009a, Schilcher2012, Haenggi2012, Crismani2013, Zhong2014}. It was first addressed in \cite{Ganti2009a} for random wireless networks with PPP distributed nodes, where the spatio-temporal correlation coefficient was introduced. The temporal correlation coefficient was evaluated in \cite{Schilcher2012} for general network models, including the static and random node locations for various traffic types. The diversity loss of a multi-antenna receiver due to interference correlation was analyzed in \cite{Haenggi2012}. For cooperative relaying, the interference correlation occurs between different receivers that are closely located to each other. In \cite{Crismani2013}, the authors assumed the PPP model for the interfering nodes and proved that the temporal and spatial correlation of interference significantly degrades the performance of the cooperative relay.

One effective method to reduce the interference correlation is to intentionally induce man-made randomness by using random medium access, \textit{i.e.}, increasing randomness in the MAC domain, more specifically using frequency-hopping multiple access (FHMA) and ALOHA, which helps to reduce the effect of interference correlation. In ALOHA, each node transmits with a certain probability $p$. Decreasing the transmit probability increases the uncertainty in the active interfering nodes and reduces the interference, thereby reducing correlation. In \cite{Zhong2014}, the authors analyzed the local delay, which is the time it takes for a node to successfully transmit to a nearby neighbor, using FHMA and ALOHA on PPP and determined the optimal number of sub-bands in FHMA and the optimal transmit probability in ALOHA that minimizes the local delay. Stamatiou and Haenggi \cite{Stamatiou2013} applied this approach to a multi-hop relay network and evaluated the local delay of time division multiple access (TDMA) and ALOHA protocols. However, the authors assumed the nodes to be aligned on a one dimensional straight line which limits the application of their work.

To the best of our knowledge, with the exception of \cite{Stamatiou2013}, there is no previous work that has considered the effect of the MAC protocol on interference correlation for a cooperative relay network. This motivates us to consider a more realistic relaying network where the nodes are spread over a two dimensional vector space and to evaluate the effect of interference correlation on the communication performance. In particular, for the first time, we compare the successful transmission probability and local delay of a relay network, which are subject to both correlated and uncorrelated interference. 
We also propose an efficient optimization strategy that jointly maximizes the success probability and minimizes the local delay. Specifically, we provide the following theoretical contributions.
\begin{enumerate}
 \item We analyze the successful transmission probability and local delay of a relay network with PPP interfering nodes. We consider both interference correlated and uncorrelated cases.
 \item We determine the necessary and sufficient condition to jointly maximize the success probability and minimize the local delay, and derive the optimal transmission probability $p$ that achieves this optimality.
 \item We propose an optimization strategy that numerically finds the optimal $p$, and then compare the computational complexity of our proposed strategy to that of the conventional brute-force method.
  \item We provide numerical results to validate the analysis, evaluate the performance of the interference correlated and uncorrelated cases, and compare the performance gain that our proposed optimization strategy achieves compared to the conventional scheme. 
\end{enumerate}

The rest of this paper is organized as follows. 
In Section II, we describe the system model and introduce the mathematical background for the interference analysis performed here. 
We derive the successful transmission probability and local delay for the interference correlated and uncorrelated cases in Section III.
In Section IV, we determine the necessary and sufficient condition to achieve the joint optimality between the success probability and local delay, and propose an optimization strategy based on iterative numerical search. Section V provides some numerical results based on our approach. Finally, Section VI concludes the paper with some closing remarks.

\section{System and Mathematical Models}

\begin{figure}[!htbp]
    \centering
    \includegraphics[width=0.4\textwidth]{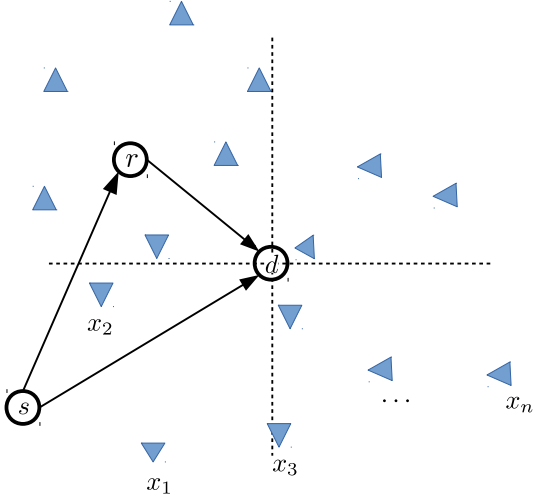} 
    \caption{System Model}
    \label{fig.system_model}
\end{figure}

\subsection{System Model}

We consider a three node relay network, consisting of a source, relay, and destination, where multiple interfering nodes simultaneously transmit during each time slot, as illustrated in \figref{fig.system_model}. We denote the source, relay, destination, and the interfering node as $s$, $r$, $d$, and $x$, respectively, where the notations denote both the nodes and their coordinates. We assume that the interfering nodes are randomly distributed according to a Poisson point process (PPP) $\Phi$ of intensity $\lambda$, \textit{i.e.}, $x \in \Phi$, and the destination is located at the origin $d = (0,0)$. The source, relay, and interfering nodes transmit with power $P_s$, $P_r$, and $P_x$, respectively.

We use the ALOHA protocol with transmit probability $p$ on each time slot, \textit{i.e.}, 
the source attempts to access a slot with success probability $p$.
Given the channel access, the source transmits its packet to the relay. 
We follow the approach of \cite{Crismani2013} and \cite{Liu2007} by assuming that a certain time is reserved for the relay transmission
right after the source transmission, so that the relay forwards the packet to the destination within the same time slot. 
If the relayed transmission fails, the source attempts to re-transmit during the next time slot. 
We assume that the link between source and destination is unreliable, so the transmission occurs only through the relay, \textit{i.e.}, there is no direct link between $s \rightarrow d$.

The distance between arbitrary node $i$ and $j$ is denoted by $||i-j||$ 
and the path loss function between two nodes is given by $l(i,j) = ||i-j||^{-\alpha}$, 
where $\alpha > 2$ is the path loss exponent. 
The channel links are assumed to be subject to independent and identically distributed (i.i.d.) Rayleigh fading with mean one, 
where the channel coefficient between node $i$ and $j$ is denoted by $h_{i j}$. The additive noise $w$ is assumed to be complex Gaussian distributed 
with mean zero and power spectral density $N_0$.

\subsection{Mathematical Model}

Let $\Phi_k$ denote the set of active interfering nodes in time slot $k$, \textit{i.e.}, $\Phi_k \in \Phi$.
The aggregated interference at the destination during time slot $k$ is 
\begin{equation}
	\begin{split}
	I_{k, d} = P_x \sum_{x \in \Phi \backslash \{s\}} h_{x d} l(x,d) \textbf{1}(x \in \Phi_k),
	\end{split}
	\label{gene.eq_1}
\end{equation}
and the interference at the relay during time $k$ is given by 
\begin{equation}
	\begin{split}
	I_{k, r} = P_x \sum_{x \in \Phi \backslash \{s\}} h_{xr} l(x,r) \textbf{1}(x \in \Phi_k),
	\end{split}
	\label{gene.eq_2}
\end{equation}
where $\textbf{1}(\cdot)$ is the indicator function. Then, the signal to interference plus noise ratio (SINR) between $s \rightarrow r$ and $r \rightarrow d$ links during time slot $k$ is
\begin{equation}
	\begin{split}
	\gamma_{k, sr} &= \frac{P_s h_{sr} l(s,r)}{N_0 + I_{k, r}} = \frac{\widehat{P}_s h_{sr} l(s,r)}{\widehat{N} + I_{k, r}^{'}},\\
	\gamma_{k, rd} &= \frac{P_r h_{rd} l(r,d)}{N_0 + I_{k, d}} = \frac{\widehat{P}_r h_{rd} l(r,d)}{\widehat{N} + I_{k, d}^{'}},
	\end{split}
	\label{gene.eq_3}
\end{equation}
where $\widehat{P_s} \triangleq P_s/P_x$, $\widehat{P_r} \triangleq P_r/P_x$, 
$\widehat{N} \triangleq N_0/P_x$, 
$I_i^{'} \triangleq I_i/P_x$.

In this paper, we consider the cases of both correlated and uncorrelated interference. For correlated interference,
we assume that the interference at the destination and the relay originate from the same set of interfering nodes,
\textit{i.e.}, $\Phi_k$ in (\ref{gene.eq_1}) and (\ref{gene.eq_2}). For uncorrelated interference, the interference at the destination 
and the relay are generated by two different sets of interfering nodes, 
\textit{i.e.}, $\Phi_k \neq \Phi_k^{'}$, such that the interference model in (\ref{gene.eq_1}), (\ref{gene.eq_2}) must be modified. 
The uncorrelated interference case is analyzed in detail in Section III-C.

\section{Analysis of the System Measures}

\subsection{Packet Delivery Probability}

The packet delivery probability, \textit{i.e.}, the successful transmission probability between $s \rightarrow r \rightarrow d$, 
is now derived in Theorem 1.

\begin{thm}
	The packet delivery probability $P(\mathcal{C})$ of the three node relay network which is subject to correlated interference is given by
	\begin{equation}
		\begin{split}
		P(\mathcal{C}) &= p \exp\left( - \lambda p \psi(r) - B \right),
		\end{split}
		\label{gene.eq_10}
	\end{equation}
	where $\psi(r)$ is defined in (\ref{gene.eq_11}), 
	$B \triangleq \widehat{N} \left( \frac{\theta_{sr}}{\widehat{P}_s} + 
	\frac{\theta_{rd}}{\widehat{P}_r} \right)$, and 
	$\theta_{ij} \triangleq \theta/l(i, j) = \theta ||i-j||^{\alpha}$.
\end{thm}

\begin{proof}
See Appendix I. 
\end{proof}

\subsection{Mean Local Delay}

The source attempts to re-transmit if the destination fails to receive the packet during the previous time slot. In general, a successful transmission during each time slot is a dependent event due to the correlated interference \cite{Zhong2014}. However, if we consider the conditional success event for a given $\Phi$, the randomness stems only from the channel fading coefficient and the ALOHA protocol which are independent variables for each time slot. Therefore, the success event in different time slots given $\Phi$ are independent with probability $P\left( \mathcal{C}_{\Phi} \right)$ in (\ref{gene.eq_5}). 

Let us define the local delay as the number of time slots required until a successful transmission occurs.
Then, the local delay given $\Phi$, denoted by $\Delta_{\Phi}$, can be represented as a geometric random variable written below
\begin{equation}
	\begin{split}
		P\left( \Delta_{\Phi} = k \right) = \left( 1 - P\left( \mathcal{C}_{\Phi} \right) \right)^{k-1} P\left( \mathcal{C}_{\Phi} \right),
		\quad
		k \geq 1.
	 \end{split}
	\label{gene.eq_10_1}
\end{equation}
The mean local delay $D(p) \triangleq \mathbb{E}\left[ \mathbb{E}\left( \Delta_{\Phi} \right) \right]$ averaged over all possible $\Phi$ is now derived in Theorem 2 below.

\begin{thm}
	The mean local delay $D(p)$ of the three node relay network for correlated interference is given by
	\begin{equation}
		\begin{split}
		D(p) 
		&= \frac{1}{p} \exp\left( \lambda p \varphi(r) + B \right),
		\end{split}
		\label{gene.eq_13}
	\end{equation}
	where $\varphi(r)$ is defined in (\ref{gene.eq_12}) and $B \triangleq \widehat{N} \left( \frac{\theta_{sr}}{\widehat{P}_s} + 
	\frac{\theta_{rd}}{\widehat{P}_r} \right)$.
\end{thm}

\begin{proof}
See Appendix II. 
\end{proof}

\begin{figure*}[!htbp]
    \begin{equation}
    \begin{split}
    \psi(r) &\triangleq  \int_{R^2} \left[ 1 - \frac{1}{\left( 1 + \frac{\theta_{sr}l(x,r)}{\widehat{P}_s} \right)
    \left( 1 + \frac{\theta_{rd}l(x,d)}{\widehat{P}_r} \right)
    }\right] dx
    = \int_{R^2} \left[ 1 - \frac{1}{\left( 1 + \frac{\theta_{sr}||x-r||^{-\alpha}}{\widehat{P}_s} \right)
    \left( 1 + \frac{\theta_{rd} ||x||^{-\alpha}}{\widehat{P}_r} \right)
    }\right] dx,
      \end{split}
        \label{gene.eq_11}
    \end{equation}

      \begin{equation}
       \begin{split}
    \varphi(r) &\triangleq  \int_{R^2} \frac{f(x)}{1 + (1-p)f(x)}dx, \quad 
      f(x) \triangleq \left( 1 + \frac{\theta_{sr}l(x,r)}{\widehat{P}_s} \right)
    \left( 1 + \frac{\theta_{rd}l(x,d)}{\widehat{P}_r} \right) - 1,
      \end{split}
        \label{gene.eq_12}
    \end{equation}

      \begin{equation}
    \begin{split}
    \psi_{u}(r) &\triangleq  \int_{R^2} \left[ 1 - \frac{1}{\left( 1 + \frac{\theta_{sr}l(x,r)}{\widehat{P}_s} \right)
    }\right] dx
    = \int_{R^2} \frac{1}{\frac{\widehat{P}_s}{\theta_{sr}} ||x-r||^{\alpha} + 1}dx,
      \end{split}
        \label{gene.eq_22}
    \end{equation}

      \begin{equation}
       \begin{split}
    \varphi_{u}(r) &\triangleq  \int_{R^2} \frac{g(x)}{1 + (1-p)g(x)}dx, \quad 
          g(x) \triangleq \frac{\theta_{sr}l(x,r)}{\widehat{P}_s},
          \end{split}
        \label{gene.eq_23}
    \end{equation}

  \hrule
    \begin{equation}
      \begin{split}
        \mathbb{E}_{I^{'}}\left[ e^{-\left( \frac{\theta_{sr}}{\widehat{P}_s} I_{k, r}^{'} + \frac{\theta_{rd}}{\widehat{P}_r}I_{k, d}^{'}  \right) } \right] 
      &= \prod_{x \in \Phi \backslash \{s\}} \mathbb{E}\left[  
      e^{-\frac{\theta_{sr}l(x,r)}{\widehat{P}_s} h_{xr} \textbf{1}(x \in \Phi_k) } \right]
      \prod_{x \in \Phi \backslash \{s\}} \mathbb{E}\left[  
      e^{-\frac{\theta_{rd}l(x,d)}{\widehat{P}_r} h_{xd} \textbf{1}(x \in \Phi_k^{'}) } \right] \\
      &= \prod_{x \in \Phi \backslash \{s\}} 
          \left[ \frac{p}{1 + \frac{\theta_{sr}l(x,r)}{\widehat{P}_s}} 
          + 1 - p \right]
          \prod_{x \in \Phi \backslash \{s\}} 
      \left[ \frac{p}{1 + \frac{\theta_{rd}l(x,d)}{\widehat{P}_r}}
      + 1 - p \right],
      \end{split}
          \label{gene.eq_19}
    \end{equation}

  \hrule  
\end{figure*}

\subsection{Results for Uncorrelated Interference}

For uncorrelated interference, the interference at the destination and the relay are generated by two different sets of nodes:
the aggregated interference at $d$ and $r$ are given by 
\begin{equation}
	\begin{split}
	I_{k, d} &= P_x \sum_{x \in \Phi \backslash \{s\}} h_{xd} l(x,d) \textbf{1}(x \in \Phi_k),\\
	I_{k, r} &= P_x \sum_{x \in \Phi \backslash \{s\}} h_{xr} l(x,r) \textbf{1}(x \in \Phi_k^{'}),
	\end{split}
	\label{gene.eq_20}
\end{equation}
where $\Phi_k$ and $\Phi_k^{'}$ are two different sets of active interfering nodes, \textit{i.e.}, $\Phi_k \neq \Phi_k^{'}$.
Then, the expectation in the last equality of (\ref{gene.eq_5}) is evaluated in (\ref{gene.eq_19})
and the corresponding packet delivery probability and the mean local delay for independent interference are derived in Theorem 3 below.

\begin{thm}
	The packet delivery probability $P(\mathcal{C})$ and the mean local delay $D(p)$ of the three node relay network 
	for uncorrelated interference are given by 
	\begin{equation}
		\begin{split}
		\resizebox{.89\hsize}{!}{$
		P(\mathcal{C}) = p \exp\left( - \lambda p \psi_{u}(r) - \lambda \pi p C(\delta) 
		\left( \frac{\theta_{rd}}{\widehat{P}_r}\right)^{\delta} 
		- B \right),
		$}
		\end{split}
		\label{gene.eq_21}
	\end{equation}
	\begin{equation}
		\resizebox{.89\hsize}{!}{$
		D(p) = \frac{1}{p} \exp\left( \lambda p \varphi_{u}(r) + \frac{\lambda \pi p C(\delta)}{(1-p)^{1-\delta}} \left( \frac{\theta_{rd}}{\widehat{P}_r}\right)^{\delta} 
		+ B \right),
		$}
		\label{gene.eq_24}
	\end{equation}
	where $\psi_u(r)$, $\varphi_{u}(r)$, and $B$ are defined in (\ref{gene.eq_22}), (\ref{gene.eq_23}), and Theorem 1, respectively.
\end{thm}

\begin{proof}
See Appendix III. 
\end{proof}

\begin{remark}
Without a MAC coordination, \textit{i.e.}, $p = 1$, the mean local delay becomes infinity as indicated in \normalfont{\cite{Zhong2014}}.
We observe the same result on cooperative relay networks. By using (\ref{gene.eq_13}), the mean local delay for correlated interference 
with $p = 1$ case is given by
	\begin{equation}
		\begin{split}
		D(p) &= \exp\left( \lambda \left. \varphi(r)\right\vert_{p=1} + B \right).
		\end{split}
		\label{neo.eq_1}
	\end{equation}
  $\left.\varphi(r)\right\vert_{p=1}$ is lower bounded as follows 
  \begin{equation}
    \begin{split}
      \left. \varphi(r)\right\vert_{p=1} &= 
      \int_{R^2} f(x)dx \geq \frac{\theta_{rd}}{\widehat{P}_r} \int_{R^2} l(x,d) dx\\
      &= \frac{\theta_{rd}}{\widehat{P}_r} \int_{R^2} ||x||^{-\alpha} dx = \frac{2 \pi\theta_{rd}}{\widehat{P}_r} \int_{r=0}^{\infty} r^{1-\alpha} dr,
    \label{neo.eq_2}
    \end{split}
  \end{equation}
  where we applied (\ref{gene.eq_12}) in the first equality and transformed the Cartesian coordinates to Polar coordinates, 
  \textit{i.e.}, $x \rightarrow re^{iw}$, in the last equality.
  Then, $\varphi(r)$ diverges to infinity at $p=1$. Similarly, for uncorrelated interference, the term within the exponential of 
  (\ref{gene.eq_24}) diverges to infinity. Hence, the mean local delay of the three node relay network becomes infinity at $p = 1$.
\end{remark}

\section{Optimization Strategy}

In this section, we determine the ALOHA transmission probability $p$ 
that optimizes the system measure in Section III. 
We determine the necessary and sufficient condition for optimality, 
propose an iterative method to find the optimal $p^{\ast}$ that meets both conditions, 
and compare the computational complexity of our proposed optimization strategy to that of 
the brute-force search method. 

\subsection{Minimizing the Delay}

We minimize the mean local delay $D(p)$ as follows
\begin{equation}
	\begin{split}
	\underset{0 \leq p \leq 1}{\text{Find}} ~ p^{\ast} ~ \text{that minimizes}~ D(p), ~\textit{i.e.,~} \frac{\partial D(p)}{\partial p} = 0.
	 \end{split}
	\label{gene.eq_30}
\end{equation}
\begin{corollary}
The optimal transmission probability $p^{\ast}$ that minimizes $D(p)$ is the solution of the following condition 
  \begin{equation}
    \begin{split}
    \frac{1}{\lambda p} = \varphi(r) + p \frac{\partial \varphi(r)}{\partial p}
    \end{split}
	    \label{gene.eq_31_a}
  \end{equation}
for correlated interference, and 
  \begin{equation}
    \begin{split}
    \resizebox{.89\hsize}{!}{$
    \frac{1}{\lambda p} = \varphi_u(r) + p \frac{\partial \varphi_u(r)}{\partial p} + 
    \frac{\pi C(\delta) \left(1-p \delta \right)}{(1-p)^{2-\delta}} \left( \frac{\theta_{rd}}{\widehat{P}_r}\right)^{\delta} 
    $}
    \end{split}
	    \label{gene.eq_31_b}
  \end{equation}
for uncorrelated interference. 
\end{corollary}

\begin{proof}
For brevity, we only prove the interference correlated case, however the proof for uncorrelated interference can be found using a similar approach.
By using Theorem 2, the first derivative of $D(p)$ is obtained as 
   \begin{equation}
    \begin{split}
    \frac{\partial D(p)}{\partial p} &= -\frac{D(p)}{p} \left[ 1 - \lambda p \left( \varphi(r) + p \frac{\partial \varphi(r)}{\partial p} \right)
    \right]. 
    \end{split}
	    \label{gene.eq_32}
  \end{equation}
 Since $D(p)$ has a positive value, the optimal $p^{\ast}$ achieves (\ref{gene.eq_31_a}). 

The second derivative of $D(p)$ at $p^{\ast}$ is given by
   \begin{equation}
    \begin{split}
        \resizebox{.89\hsize}{!}{$
    \frac{\partial^2 D(p)}{\partial p^2}\bigg|_{p^{\ast}} = \frac{D(p)}{p^2} 
    \left[ 1 + \lambda p^2 \left( 2 \frac{\partial \varphi(r)}{\partial p} + p \frac{\partial^2 \varphi(r)}{\partial p^2} \right)
    \right], 
    $}
    \end{split}
	    \label{gene.neoeq_01}
  \end{equation}
where we applied (\ref{gene.eq_31_a}) to simplify the expression.
The derivatives of $\varphi(r)$ are given by 
  \begin{equation}
    \begin{split}
    \frac{\partial \varphi(r)}{\partial p} &=  \int_{R^2} \frac{f(x)^2}{(1 + (1-p)f(x))^2}dx, \\
    \frac{\partial^2 \varphi(r)}{\partial p^2} &=  \int_{R^2} \frac{f(x)^3}{(1 + (1-p)f(x))^3}dx.
    \end{split}
	    \label{gene.neoeq_02}
  \end{equation}
 Since $f(x)$ is a non-negative function on $x$, 
 $\frac{\partial \varphi(r)}{\partial p}$ and $\frac{\partial^2 \varphi(r)}{\partial p^2}$ are both non-negative.
 Hence, $\frac{\partial^2 D(p)}{\partial p^2}$ in (\ref{gene.neoeq_01}) has a positive value at $p^{\ast}$ 
 and the transmission probability $p^{\ast}$ minimizes the mean local delay $D(p)$. This completes the proof. 
\end{proof}

\subsection{Joint Optimization of the Throughput and Delay}

Let us define the utility function $U(p) \triangleq \frac{p P(\mathcal{C})}{D(p)}$ as the ratio of 
the network throughput $p P(\mathcal{C})$ to the mean local delay $D(p)$. By maximizing $U$, 
we can jointly maximize the throughput and minimize the delay as follows
\begin{equation}
	\begin{split}
	\underset{0 \leq p \leq 1}{\text{Find}} ~ p^{\ast} ~ \text{that maximizes}~ U, ~\textit{i.e.,~} \frac{\partial U}{\partial p} = 0.
	 \end{split}
	\label{gene.eq_33}
\end{equation}
\begin{corollary}
The optimal probability $p^{\ast}$ that maximizes the utility $U$ is the solution of the following condition 
  \begin{equation}
    \begin{split}
    \frac{3}{\lambda p} = \psi(r) + \varphi(r) + p \frac{\partial \varphi(r)}{\partial p}
    \end{split}
	    \label{gene.neoeq_03}
  \end{equation}
for correlated interference, and 
  \begin{equation}
    \begin{split}
    \frac{3}{\lambda p} &= \psi_u(r) + \varphi_u(r) + p \frac{\partial \varphi_u(r)}{\partial p}\\
    &+ \pi C(\delta) \left( \frac{\theta_{rd}}{\widehat{P}_r}\right)^{\delta} \left[ 1 + \frac{(1-p \delta)}{(1-p)^{2-\delta}}\right]
    \end{split}
	    \label{gene.neoeq_04}
  \end{equation}
for uncorrelated interference. 
\end{corollary}

\begin{proof}
Again, we only prove the interference correlated case, since the uncorrelated interference case can be obtained in a similar manner.
By using Theorem 1 and 2, the first derivative of $U$ can be expressed as follows
   \begin{equation}
    \begin{split}
    \resizebox{.89\hsize}{!}{$		
    \frac{\partial U(p)}{\partial p} = \frac{U(p)}{p} \left[ 3 - 
    \lambda p \left( \psi(r) + \varphi(r) + p \frac{\partial \varphi(r)}{\partial p} \right)
    \right].
    $}
    \end{split}
	    \label{gene.eq_36}
  \end{equation}
 Since $U(p)$ has a positive value for $p > 0$, the optimal $p^{\ast}$ achieves (\ref{gene.neoeq_03}).
The second derivative of $U(p)$ at $p^{\ast}$ is given by
   \begin{equation}
    \begin{split}
        \resizebox{.89\hsize}{!}{$
    \frac{\partial^2 U(p)}{\partial p^2}\bigg|_{p^{\ast}} = -\frac{U(p)}{p^2} 
    \left[ 3 + \lambda p^2 \left( 2 \frac{\partial \varphi(r)}{\partial p} + p \frac{\partial^2 \varphi(r)}{\partial p^2} \right)
    \right], 
    $}
    \end{split}
	    \label{gene.neoeq_05}
  \end{equation}
where we applied (\ref{gene.neoeq_03}) to simplify the expression. 
Since $\frac{\partial \varphi(r)}{\partial p}$ and $\frac{\partial^2 \varphi(r)}{\partial p^2}$ are both non-negative, 
$\frac{\partial^2 U(p)}{\partial p^2}$ in (\ref{gene.neoeq_05}) has a negative value at $p^{\ast}$ 
 and the transmission probability $p^{\ast}$ maximizes the utility $U(p)$. This completes the proof. 
 \end{proof}

\subsection{Proposed Optimization Strategy}

\begin{algorithm}
  \caption{Find $p^{\ast}$ that achieves (\ref{gene.eq_31_a}), (\ref{gene.eq_31_b}), (\ref{gene.neoeq_03}), (\ref{gene.neoeq_04}).}
  \begin{algorithmic}[1]
    \Require $0 < p_0 \leq 1, \epsilon_0 > 0$
		\Procedure{Newton-Raphson Iteration}{}
		 \State $p \gets p_0$
      \While{$|\Lambda(p)| > \epsilon_0$}
        \State $p \gets p - \frac{\Lambda(p)}{\Lambda^{'}(p)}$
      \EndWhile\label{euclidendwhile}
      \State \textbf{return} $p^{\ast} \gets p$
    \EndProcedure
  \end{algorithmic}
	\label{opt.strategy}
\end{algorithm}

The optimal conditions in Corollary 1 and 2 are implicit functions of $p$: the effect of the interference and the ALOHA transmission probability $p$ are coupled inside the integral term of $\varphi(r)$ and $\varphi_u(r)$, which makes it hard to explicitly calculate $p^{\ast}$ from (\ref{gene.eq_31_a}), (\ref{gene.eq_31_b}), (\ref{gene.neoeq_03}), and (\ref{gene.neoeq_04}). Hence, we adopt an iterative search method to numerically find the optimal 
$p^{\ast}$. Specifically, we use Newton-Raphson iteration as outlined in Algorithm 1 \cite{MTheath2002}. 
First, we denote $\Lambda(p)$ for (\ref{gene.neoeq_03}) as 
   \begin{equation}
    \begin{split}
    \Lambda(p) = \frac{3}{\lambda p} - \psi(r) - \varphi(r) - p \frac{\partial \varphi(r)}{\partial p}.
    \end{split}
	    \label{gene.neoeq_06}
  \end{equation}
 Similarly, $\Lambda(p)$ for (\ref{gene.eq_31_a}), (\ref{gene.eq_31_b}), 
  and (\ref{gene.neoeq_04}) are defined by subtracting the terms on the left-hand side to that on the right-hand side. 
  The derivative of $\Lambda(p)$ for (\ref{gene.neoeq_03}) is given by 
   \begin{equation}
    \begin{split}
    \Lambda^{'}(p) = -\frac{3}{\lambda p^2} - 2 \frac{\partial \varphi(r)}{\partial p} - p\frac{\partial^2 \varphi(r)}{\partial p^2},
    \end{split}
	    \label{gene.neoeq_07}
  \end{equation}
where the derivatives of $\varphi(r)$ are derived in (\ref{gene.neoeq_02}).

The iterative search begins by setting the initial transmission probability $0 < p_0 \leq 1$ and the termination criteria $\epsilon_0 > 0$. 
The ALOHA transmission probability is updated as  
   \begin{equation}
    \begin{split}
    	p_{m+1} = p_m - \frac{\Lambda(p_m)}{\Lambda^{'}(p_m)}, \quad m:\text{ iteration index}, 
    \end{split}
	    \label{gene.neoeq_08}
  \end{equation}
until $\Lambda(p_m)$ becomes arbitrarily close to zero. If the transmission probability $p_m$ satisfies the condition 
$|\Lambda(p_m)| \leq \epsilon_0$, the iterative search stops and returns $p^{\ast} = p_m$ as the optimal transmission probability.

Conventional approaches \cite{MTheath2002} for numerically finding $p^{\ast}$ include
the brute-force method where the utility $U(p)$ is evaluated for all possible $0 < p \leq 1$ and compared to find the optimal value 
and the quick-sort method where each adjacent numbers $U(p_1)$ and $U(p_2)$ are compared and sorted until the maximum value is reached. 
The ratio between the computational complexity of the Newton-Raphson iteration and that of the brute-force method is 
${\mathcal{O}(\log{}n)}/{\mathcal{O}(n^2)}$ for $n$ digit precision, whereas the ratio between the computational complexity of the Newton-Raphson iteration and that of the quick-sort method is ${\mathcal{O}(\log{}n)}/{\mathcal{O}(n\log{}n)}$ \cite{MTheath2002}. Hence, the optimization 
strategy proposed here is computationally efficient and effective compared to the conventional schemes, such as brute-force method or quick-sort method.

\section{Numerical Results}

In this section, we present numerical examples through which we compare the performance of the interference correlated and uncorrelated cases. 
\figref{fig1} shows the success probability, mean local delay, and utility calculated using 
(\ref{gene.eq_10}), (\ref{gene.eq_13}), (\ref{gene.eq_21}), (\ref{gene.eq_24}), and $U(p) = \frac{p P(\mathcal{C})}{D(p)}$
and evaluated over different relay locations. We use a three node relay network where the nodes are located at $s = (2, 0)$, $d = (0, 0)$, 
$r = (r_x, 0)$. The interfering nodes are randomly deployed by using PPP with node density $\lambda$. 
We assume that the nodes transmit with power $P_s = 5~\text{dB}$, $P_r = 5~\text{dB}$, $P_x = 5~\text{dB}$, the noise power is $N_0 = 1$, 
the path loss exponent is $\alpha = 4$, and the SINR threshold is $\theta = 1$. 
The dotted and solid curves are the numerical results for the interference correlated (IC) and interference uncorrelated (IU) cases, respectively. 
It is clear that the IC achieves a better performance than the IU in terms of success probability, mean local delay, and utility. 
For a small interfering node density $\lambda$, the performance gap between the IC and IU cases is very small, however, 
as the node density increases, the gap becomes more evident. 
We also note that for a small node density, using a high transmission probability achieves a better performance 
in terms of the success probability, mean local delay, and utility. For a large node density, more nodes will interfere with each other, 
increasing the mean local delay and decreasing the success probability. So, it is optimal to use a low transmission probability for large node density case. 

\figref{fig2} shows the success probability versus node density, transmission probability, and SINR threshold 
where we fixed the relay location at the middle, \textit{i.e.}, $r = (1, 0)$. 
The dotted and solid curves are again the numerical results for the IC and IU cases with a fixed transmission probability $p$.
For the curves without a line and only the markers, we determined the optimal $p^{\ast}$ that maximizes the utility $U(p)$ using Corollary 2 for each node density $\lambda$ and SINR threshold $\theta$. Then, we evaluated the corresponding success probability for the given $(p^{\ast}, \lambda, \theta)$. It is clear that the IC case achieves better performance than the IU case and 
the performance gap increases as the node density increases. By using similar method that was used in Theorem 1 and 3, 
the link success probability can be derived as \cite{Haenggi2013}
   \begin{equation}
    \begin{split}
    	P\left[ \gamma_{k, rd} > \theta \right] &= \exp\left(-\frac{\widehat{N} \theta_{rd}}{\widehat{P}_r} 
        - \lambda \pi p C(\delta) \left( \frac{\theta_{rd}}{\widehat{P}_r}\right)^{\delta}  \right), \\
		P\left[ \gamma_{k, sr} > \theta \right] &= \exp\left(-\frac{\widehat{N} \theta_{sr}}{\widehat{P}_s} 
        - \lambda \pi  \psi_{u}(r)  \right),
    \end{split}
	    \label{neoeq.eq_001}
  \end{equation}
which are decreasing function of the SINR threshold $\theta$. So, each link falls in outage with a high probability
for large $\theta$ and the gap between IC and IU becomes increasingly wider as the node density increases. 
Since we jointly optimized the network throughput $p P(\mathcal{C})$ and the mean local delay $D(p)$, 
the optimized $p$ do not achieve the maximum success probability for all cases. However, 
our proposed optimization strategy still achieves a significant performance gain, compared to the conventional fixed transmission probability case. 

Similarly, \figref{fig3} shows the mean local delay versus node density, transmission probability, and SINR threshold 
where the relay location was fixed at the middle, \textit{i.e.}, $r = (1, 0)$. 
The IC case achieves a lower mean local delay than the IU case in most of the scenarios and the performance gap increases for a large node density. 
Fig. 4(b) illustrates the convexity of the mean local delay, \textit{i.e.}, as the node density increases, 
the transmission probability that minimizes the mean local delay decreases. 
Again, although the optimized $p$ determined by Corollary 2 does not achieve the minimum mean local delay for all cases, 
it still achieves a significant performance gain compared to the conventional fixed transmission probability case.

\section{Conclusion}

In this paper, we have considered a three node relay network subject to interference generated by nodes which were distributed according to a PPP. 
Furthermore, we evaluated the packet delivery probability and the mean local delay for both interference correlated and uncorrelated cases. 
Based on these analytical derivations, we determined the necessary and sufficient condition to jointly maximize the packet delivery probability and minimize the mean local delay. The transmission probability $p$ that achieves this optimality was also found. 
Specifically, we observed that correlated interference achieves better performance than uncorrelated interference in terms of success probability, mean local delay, and utility. For a small node density, the performance gap between IC and IU is very small, however, as the node density increases, the gap becomes more evident. Also, for a small node density, using a high transmission probability $p$ achieves better performance. 
However for a large node density the opposite is true, \textit{i.e.}, a low transmission probability should be used. 
Based on these observations, we have proposed a computationally efficient optimization strategy that finds the optimal $p$ using an iterative search. Finally, we have also provided numerical results to prove the performance gain that our proposed optimization strategy achieves compared to a conventional scheme.




\section*{Appendix I}

In this appendix, we prove Theorem $1$. 
	Let $\mathcal{C}_{\Phi}$ denote the successful transmission event between $s \rightarrow r \rightarrow d$ conditioned upon the PPP $\Phi$. 
	Using (\ref{gene.eq_3}) and a pre-defined SINR threshold $\theta$, $\mathcal{C}_{\Phi}$ can be written as 
	\begin{equation}
		\begin{split}
			\mathcal{C}_{\Phi} &\triangleq \{ \gamma_{k, sr} > \theta ~\cap~ \gamma_{k, rd} > \theta  \},
		\end{split}
		\label{gene.eq_4}
	\end{equation}
	and the conditional success probability given $\Phi$ is 
	\begin{equation}
		\begin{split}
			P(\mathcal{C}_{\Phi}) &= p P\left[ \gamma_{k, sr} > \theta \cap \gamma_{k, rd} > \theta | \Phi \right]\\
			&=pP\left[ \widehat{P}_s h_{sr} > \theta_{sr} ( \widehat{N} + I_{k, r}^{'} ) \right. \\
			&\left. \qquad \cap ~
			\widehat{P}_r h_{rd} > \theta_{rd} ( \widehat{N} + I_{k, d}^{'} ) | \Phi \right] \\
			&= p\mathbb{E}_{I^{'}}\left[ e^{-\frac{\theta_{sr}}{\widehat{P}_s} ( \widehat{N} + I_{k, r}^{'} ) }
			e^{-\frac{\theta_{rd}}{\widehat{P}_r} ( \widehat{N} + I_{k, d}^{'} ) }
			\right] \\
			&= p e^{-\widehat{N} \left( \frac{\theta_{sr}}{\widehat{P}_s} + \frac{\theta_{rd}}{\widehat{P}_r} \right)}
			\mathbb{E}_{I^{'}}\left[ e^{-\frac{\theta_{sr}}{\widehat{P}_s} I_{k, r}^{'}-\frac{\theta_{rd}}{\widehat{P}_r}I_{k, d}^{'} }
			\right],
		\end{split}
		\label{gene.eq_5}
	\end{equation}
	where the first equality is due to the ALOHA transmission probability $p$ 
	and the distribution of $h_{ij}$, \textit{i.e.}, $P(h_{ij} > x) = \exp(-x)$, is applied to the third equality.

	For correlated interference\footnote{
	In Section III-A and B, we only consider correlated interference. The results for uncorrelated inference are summarized in Section III-C.
	}, the interference is generated from the same set $\Phi_k$. By substituting (\ref{gene.eq_1}) and (\ref{gene.eq_2}) 
	in the last equality of (\ref{gene.eq_5}) and averaging it with respect to the fading coefficient $h_{xi}$ and the ALOHA protocol $\Phi_k$, 
	the expectation in the last term of (\ref{gene.eq_5}) can be expressed as 
	\begin{equation}
		\begin{split}
			&\mathbb{E}_{I^{'}}\left[ e^{-\left( \frac{\theta_{sr}}{\widehat{P}_s} I_{k, r}^{'} + \frac{\theta_{rd}}{\widehat{P}_r}I_{k, d}^{'}  \right) } \right] \\
			= &\prod_{x \in \Phi \backslash \{s\}} \mathbb{E}\left[  
			e^{-\left( \frac{\theta_{sr}l(x,r)}{\widehat{P}_s} h_{xr} + \frac{\theta_{rd}l(x,d)}{\widehat{P}_r}h_{xd} \right)\textbf{1}(x \in \Phi_k) } \right] \\
			= &\prod_{x \in \Phi \backslash \{s\}} 
			\left[ \frac{p}{1 + \frac{\theta_{sr}l(x,r)}{\widehat{P}_s}} \frac{1}{1 + \frac{\theta_{rd}l(x,d)}{\widehat{P}_r}}
			+ 1 - p \right].
		\end{split}
		\label{gene.eq_6}
	\end{equation}

	Using (\ref{gene.eq_5}) and (\ref{gene.eq_6}), the packet delivery probability, \textit{i.e.}, the average success probability over all possible $\Phi$, can now be expressed as follows
	\begin{equation}
		\begin{split}
			P(\mathcal{C}) &= \mathbb{E}\left[ P(\mathcal{C}_{\Phi}) \right]\\
			&= p e^{-\widehat{N} \left( \frac{\theta_{sr}}{\widehat{P}_s} + 
			\frac{\theta_{rd}}{\widehat{P}_r} \right)} \mathbb{E}\left[ \prod_{x \in \Phi \backslash \{s\}} \upsilon(x) \right],
		\end{split}
		\label{gene.eq_7}
	\end{equation}
	where 
	\begin{equation}
		\begin{split}
			\upsilon(x) &= \frac{p}{1 + \frac{\theta_{sr}l(x,r)}{\widehat{P}_s}} \frac{1}{1 + \frac{\theta_{rd}l(x,d)}{\widehat{P}_r}} + 1 - p.
		\end{split}
	\label{gene.eq_8}
	\end{equation}
	The expectation term in (\ref{gene.eq_7}) is referred to as the probability generating functional (PGFL) 
	and can be evaluated for PPP as \cite{Haenggi2013}
	\begin{equation}
		\begin{split}
			G[\upsilon(x)] &= \mathbb{E}\left[ \prod_{x \in \Phi} \upsilon(x) \right]\\ 
			&= \exp\left( - \lambda \int_{R^2} [ 1 - \upsilon(x)] dx\right),
		\end{split}
		\label{gene.eq_9}
	\end{equation}
	where the integral in (\ref{gene.eq_9}) is a two dimensional integral since $x$ is a coordinate.
	Hence, we obtain the packet delivery probability in (\ref{gene.eq_10}) by substituting (\ref{gene.eq_8}) and (\ref{gene.eq_9}) into (\ref{gene.eq_7}) and 
	using the notation $\psi(r)$ in (\ref{gene.eq_11}). This completes the proof.

\section*{Appendix II}

In this appendix, we prove Theorem $2$.
	The mean of the geometric random variable in (\ref{gene.eq_10_1}) is  $\mathbb{E}\left[ \Delta_{\Phi} \right] = 1/P(\mathcal{C}_{\Phi})$. 
	Then, the mean local delay for all available $\Phi$ is given by
	\begin{equation}
		\begin{split}
		D(p) &\triangleq \mathbb{E}\left[ \mathbb{E}\left( \Delta_{\Phi} \right) \right] = \mathbb{E}\left[ \frac{1}{P\left( \mathcal{C}_{\Phi} \right)} \right].
		\end{split}
		\label{gene.eq_14}
	\end{equation}
	By substituting (\ref{gene.eq_5}) and (\ref{gene.eq_6}) into the last equality in (\ref{gene.eq_14}), applying the PGFL of PPP (\ref{gene.eq_9}), the mean local delay can be expressed as follows
	\begin{equation}
		\begin{split}
		D(p) &= \frac{e^{B}}{p} \mathbb{E}\left[ \prod_{x \in \Phi \backslash \{s\}} \frac{1}{\upsilon(x)} \right] \\
		&= \frac{1}{p} \exp\left( B - \lambda \int_{R^2} \left[ 1 - \frac{1}{\upsilon(x)}\right] dx \right).
		\end{split}
		\label{gene.eq_14_b}
	\end{equation}
	Since the integral in (\ref{gene.eq_14_b}) can be simplified as 
	\begin{equation}
		\begin{split}
		-p\varphi(r) = \int_{R^2} \left[ 1 - \frac{1}{\upsilon(x)}\right] dx, 
		\end{split}
		\label{gene.eq_14_c}
	\end{equation}
	we obtain the mean local delay in (\ref{gene.eq_13}) by substituting (\ref{gene.eq_14_c}) into (\ref{gene.eq_14_b}). This completes the proof.

\section*{Appendix III}

In this appendix, we prove Theorem $3$.
	The conditional success probability given $\Phi$ is 
	\begin{equation}
		\begin{split}
		P(\mathcal{C}_{\Phi}) &= p e^{-B}
		\mathbb{E}_{I^{'}}\left[ e^{-\frac{\theta_{sr}}{\widehat{P}_s} I_{k, r}^{'}-\frac{\theta_{rd}}{\widehat{P}_r}I_{k, d}^{'} } \right] \\
		&= p e^{-B}
		\prod_{x \in \Phi \backslash \{s\}} \mu_1(x) \prod_{x \in \Phi \backslash \{s\}} \mu_2(x),
		\end{split}
		\label{gene.eq_25}
	\end{equation}
	where the first equality follows by (\ref{gene.eq_5}),  the second equality follows by (\ref{gene.eq_19}), 
	$\mu_1(x)$ and $\mu_2(x)$ are denoted by
	\begin{equation}
		\begin{split}
		\mu_1(x) &\triangleq \frac{p}{1 + \frac{\theta_{sr}l(x,r)}{\widehat{P}_s}} + 1 - p, \\
		\mu_2(x) &\triangleq \frac{p}{1 + \frac{\theta_{rd}l(x,d)}{\widehat{P}_r}} + 1 - p.
		\end{split}
		\label{gene.eq_26}
	\end{equation}
	By averaging (\ref{gene.eq_25}) over all possible $\Phi$, the packet delivery probability for uncorrelated interference can be expressed as
	\begin{equation}
		\begin{split}
		P(\mathcal{C}) &= p e^{-B} \mathbb{E}\left[ \prod_{x \in \Phi \backslash \{s\}} \mu_1(x) \right] 
		\mathbb{E}\left[ \prod_{x \in \Phi \backslash \{s\}} \mu_2(x) \right]\\
		&= p \exp\left( -B - \lambda \int_{R^2} [ 1 - \mu_1(x)] dx \right)\\
		&\times \exp\left( - \lambda \int_{R^2} [ 1 - \mu_2(x)] dx \right),
		\end{split}
		\label{gene.eq_27}
	\end{equation}
	where the first equality follows by (\ref{gene.eq_7}) and the PGFL of PPP (\ref{gene.eq_9}) is used for the second equality.  
	The first integral term in (\ref{gene.eq_27}) can be represented as $\int_{R^2} [ 1 - \mu_1(x)] dx = p \psi_{u}(r)$ by using straight forward 
	calculus and the second integral term can be expressed in closed form as follows 
	\begin{equation}
		\begin{split}
		&\int_{R^2} [ 1 - \mu_2(x)] dx = p \int_{R^2} \frac{1}{1 + \frac{\widehat{P}_r}{\theta_{rd}} ||x||^{\alpha}}dx\\
		&= 2 \pi p \int_{\rho = 0}^{\infty} \frac{\rho}{1 + \frac{\widehat{P}_r}{\theta_{rd}} \rho^{\alpha}}d\rho = 
		\pi p C(\delta) \left( \frac{\theta_{rd}}{\widehat{P}_r}\right)^{\delta}, 
		\end{split}
		\label{gene.eq_28}
	\end{equation}
	where we applied the notation of $\mu_2(x)$ in (\ref{gene.eq_26}) to the first equality, 
	transformed the Cartesian coordinates to Polar coordinates, \textit{i.e.}, $x \rightarrow \rho e^{i\omega}$, in the second equality, and 
	applied the following integration equality in the last inequality \cite{Gradshteyn2007}
	\begin{equation}
		\begin{split}
		\int_{0}^{\infty} \frac{x^{\mu - 1}}{1 + q x^{\nu}}dx = \frac{1}{\mu} q^{-\frac{\mu}{\nu}} C\left(\frac{\mu}{\nu}\right),
		\quad C(\delta) = \frac{1}{\mathrm{sinc}(\delta)}.
		\end{split}
		\label{gene.eq_29}
	\end{equation}
	(\ref{gene.eq_21}) follows by substituting (\ref{gene.eq_28}) and (\ref{gene.eq_22}) into (\ref{gene.eq_27}). 

	Similarly, the mean local delay can be derived as follows
	\begin{equation}
		\begin{split}
		D(p) &= \frac{e^{B}}{p} \mathbb{E}\left[ \prod_{x \in \Phi \backslash \{s\}} \frac{1}{\mu_1(x)} \right] 
		\mathbb{E}\left[ \prod_{x \in \Phi \backslash \{s\}} \frac{1}{\mu_2(x)} \right] \\
		&= \frac{1}{p} \exp\left( B - \lambda \int_{R^2} \left[ 1 - \frac{1}{\mu_1(x)}\right] dx \right)\\
		&\times \exp\left( - \lambda \int_{R^2} \left[ 1 - \frac{1}{\mu_2(x)}\right] dx \right),
		\end{split}
		\label{gene.eq_37}
	\end{equation}
	where the two integral terms in (\ref{gene.eq_37}) can be evaluated as 
	\begin{equation}
		\begin{split}
		&\int_{R^2} \left[ 1 - \frac{1}{\mu_1(x)} \right] dx = -p \varphi_{u}(r),\\
		&\int_{R^2} \left[ 1 - \frac{1}{\mu_2(x)} \right] dx = -\frac{\pi p C(\delta)}{(1-p)^{1-\delta}} \left( \frac{\theta_{rd}}{\widehat{P}_r}\right)^{\delta},
		\end{split}
		\label{gene.eq_38}
	\end{equation}
	by using the same integration technique in (\ref{gene.eq_28}) and (\ref{gene.eq_29}). 
	Hence, we obtain the mean local delay in (\ref{gene.eq_24}) by substituting (\ref{gene.eq_38}) and (\ref{gene.eq_23})
	into (\ref{gene.eq_37}). This completes the proof.

 \bibliographystyle{IEEEtran}
 \bibliography{bib1}

\begin{thebibliography}{10}
\providecommand{\url}[1]{#1}
\csname url@samestyle\endcsname
\providecommand{\newblock}{\relax}
\providecommand{\bibinfo}[2]{#2}
\providecommand{\BIBentrySTDinterwordspacing}{\spaceskip=0pt\relax}
\providecommand{\BIBentryALTinterwordstretchfactor}{4}
\providecommand{\BIBentryALTinterwordspacing}{\spaceskip=\fontdimen2\font plus
\BIBentryALTinterwordstretchfactor\fontdimen3\font minus
  \fontdimen4\font\relax}
\providecommand{\BIBforeignlanguage}[2]{{%
\expandafter\ifx\csname l@#1\endcsname\relax
\typeout{** WARNING: IEEEtran.bst: No hyphenation pattern has been}%
\typeout{** loaded for the language `#1'. Using the pattern for}%
\typeout{** the default language instead.}%
\else
\language=\csname l@#1\endcsname
\fi
#2}}
\providecommand{\BIBdecl}{\relax}
\BIBdecl

\bibitem{Cover1979}
T.~Cover and A.~E. Gamal, ``{Capacity theorems for the relay channel},''
  \emph{IEEE Trans. Inf. Theory}, vol.~25, no.~5, pp. 572--584, 1979.

\bibitem{Sendonaris2003}
A.~Sendonaris, E.~Erkip, and B.~Aazhang, ``{User cooperation diversity - Part
  I: System description},'' \emph{IEEE Trans. Commun.}, vol.~51, no.~11, pp.
  1927--1938, 2003.

\bibitem{Laneman2004}
J.~N. Laneman, D.~N.~C. Tse, and G.~W. Wornell, ``{Cooperative diversity in
  wireless networks: Efficient protocols and outage behavior},'' \emph{IEEE
  Trans. Inf. Theory}, vol.~50, no.~12, pp. 3062--3080, 2004.

\bibitem{Laneman2003}
J.~N. Laneman and G.~W. Wornell, ``{Distributed space-time-coded protocols for
  exploiting cooperative diversity in wireless networks},'' \emph{IEEE Trans.
  Inf. Theory}, vol.~49, no.~10, pp. 2415--2425, 2003.

\bibitem{Kramer2005}
G.~Kramer, M.~Gastpar, and P.~Gupta, ``{Cooperative strategies and capacity
  theorems for relay networks},'' \emph{IEEE Trans. Inf. Theory}, vol.~51,
  no.~9, pp. 3037--3063, 2005.

\bibitem{Kim2008}
Y.-H. Kim, ``{Capacity of a class of deterministic relay channels},''
  \emph{IEEE Trans. Inf. Theory}, vol.~54, no.~3, pp. 1328--1329, 2008.

\bibitem{Aleksic2009}
M.~Aleksic, P.~Razaghi, and W.~Yu, ``{Capacity of a class of modulo-sum relay
  channels},'' \emph{IEEE Trans. Inf. Theory}, vol.~55, no.~3, pp. 921--930,
  2009.

\bibitem{Sanderovich2009}
A.~Sanderovich, O.~Somekh, H.~V. Poor, and S.~Shamai, ``{Uplink macro diversity
  of limited backhaul cellular network},'' \emph{IEEE Trans. Inf. Theory},
  vol.~55, no.~8, pp. 3457--3478, 2009.

\bibitem{Lim2011}
S.~H. Lim, Y.~H. Kim, A.~{El Gamal}, and S.~Y. Chung, ``{Noisy network
  coding},'' vol.~57, no.~5, pp. 3132--3152, 2011.

\bibitem{Nazer2011}
B.~Nazer and M.~Gastpar, ``{Compute-and-forward: Harnessing interference
  through structured codes},'' \emph{IEEE Trans. Inf. Theory}, vol.~57, no.~10,
  pp. 6463--6486, 2011.

\bibitem{Krikidis2009}
I.~Krikidis, J.~S. Thompson, S.~McLaughlin, and N.~Goertz, ``{Max-min relay
  selection for legacy amplify-and-forward systems with interference},''
  \emph{IEEE Trans. Wirel. Commun.}, vol.~8, no.~6, pp. 3016--3027, 2009.

\bibitem{Zhong2010}
C.~Zhong, S.~Jin, and K.~K. Wong, ``{Dual-hop systems with noisy relay and
  interference-limited destination},'' \emph{IEEE Trans. Commun.}, vol.~58,
  no.~3, pp. 764--768, 2010.

\bibitem{Si2010}
J.~Si, Z.~Li, and Z.~Liu, ``{Outage probability of opportunistic relaying in
  Rayleigh fading channels with multiple interferers},'' \emph{IEEE Signal
  Process. Lett.}, vol.~17, no.~5, pp. 445--448, 2010.

\bibitem{Haenggi2013}
M.~Haenggi, \emph{{Stochastic geometry for wireless networks}}.\hskip 1em plus
  0.5em minus 0.4em\relax Cambridge University Press, 2013, vol.~I.

\bibitem{Haenggi2005}
------, ``{On distances in uniformly random networks},'' \emph{IEEE Trans. Inf.
  Theory}, vol.~51, no.~10, pp. 3584--3586, 2005.

\bibitem{Mathar1995}
R.~Mathar and J.~Mattfeldt, ``{On the distribution of cumulated interference
  power in Rayleigh fading channels},'' \emph{Wirel. Networks}, vol.~1, no.~1,
  pp. 31--36, 1995.

\bibitem{Andrews2011}
J.~G. Andrews, F.~Baccelli, and R.~K. Ganti, ``{A tractable approach to
  coverage and rate in cellular networks},'' \emph{IEEE Trans. Commun.},
  vol.~59, no.~11, pp. 3122--3134, 2011.

\bibitem{Kountouris2009}
M.~Kountouris and J.~G. Andrews, ``{Throughput scaling laws for wireless ad hoc
  networks with relay selection},'' in \emph{VTC Spring 2009 - IEEE 69th Veh.
  Technol. Conf.}, 2009, pp. 1--5.

\bibitem{Ganti2009}
R.~K. Ganti and M.~Haenggi, ``{Analysis of uncoordinated opportunistic two-hop
  wireless ad hoc systems},'' in \emph{IEEE Int. Symp. Inf. Theory - Proc.},
  2009, pp. 1020--1024.

\bibitem{Cho2011}
S.~R. Cho, W.~Choi, and K.~Huang, ``{QoS provisioning relay selection in random
  relay networks},'' \emph{IEEE Trans. Veh. Technol.}, vol.~60, no.~6, pp.
  2680--2689, 2011.

\bibitem{Yu2010}
C.~H. Yu, O.~Tirkkonen, and J.~Hmlinen, ``{Opportunistic relay selection with
  cooperative macro diversity},'' \emph{Eurasip J. Wirel. Commun. Netw.}, 2010.

\bibitem{Ganti2009a}
R.~Ganti and M.~Haenggi, ``{Spatial and temporal correlation of the
  interference in ALOHA ad hoc networks},'' \emph{IEEE Commun. Lett.}, vol.~13,
  no.~9, pp. 631--633, 2009.

\bibitem{Schilcher2012}
U.~Schilcher, C.~Bettstetter, and G.~Brandner, ``{Temporal correlation of
  interference in wireless networks with Rayleigh block fading},'' \emph{IEEE
  Trans. Mob. Comput.}, vol.~11, no.~12, pp. 2109--2120, 2012.

\bibitem{Haenggi2012}
M.~Haenggi, ``{Diversity loss due to interference correlation},'' \emph{IEEE
  Commun. Lett.}, vol.~16, no.~10, pp. 1600--1603, 2012.

\bibitem{Crismani2013}
\BIBentryALTinterwordspacing
A.~Crismani, U.~Schilcher, S.~Toumpis, and C.~Bettstetter, ``{Cooperative
  relaying in wireless networks under spatially and temporally correlated
  interference},'' \emph{arXiv preprint}, 2013. [Online]. Available:
  \url{http://arxiv.org/abs/1308.0490}
\BIBentrySTDinterwordspacing

\bibitem{Zhong2014}
Y.~Zhong, W.~Zhang, and M.~Haenggi, ``{Managing interference correlation
  through random medium access},'' \emph{IEEE Trans. Wirel. Commun.}, vol.~13,
  no.~2, pp. 928--941, 2014.

\bibitem{Stamatiou2013}
K.~Stamatiou and M.~Haenggi, ``{Delay characterization of multihop transmission
  in a Poisson field of interference},'' \emph{IEEE/ACM Trans. Netw.}, vol.~22,
  pp. 1794--1807, 2013.

\bibitem{Liu2007}
P.~Liu, Z.~Tao, S.~Narayanan, T.~Korakis, and S.~S. Panwar, ``{CoopMAC: A
  cooperative MAC for wireless LANs},'' \emph{IEEE J. Sel. Areas Commun.},
  vol.~25, no.~2, pp. 340--353, 2007.

\bibitem{MTheath2002}
M.~T. Heath, \emph{{Scientific computing: An introductory survey}}.\hskip 1em
  plus 0.5em minus 0.4em\relax McGraw-Hill, 2002.

\bibitem{Gradshteyn2007}
I.~S. Gradshteyn and I.~M. Ryzhik, \emph{{Table of integrals, series, and
  products}}.\hskip 1em plus 0.5em minus 0.4em\relax Academic Press, 2007.

\bibitem{Ordentlich2014}
O.~Ordentlich, U.~Erez, and B.~Nazer, ``{The approximate sum capacity of the
  symmetric gaussian $K$-User interference channel},'' \emph{IEEE Trans. Inf.
  Theory}, vol.~60, no.~6, pp. 3450--3482, 2014.

\bibitem{Niesen2012}
U.~Niesen and P.~Whiting, ``{The degrees of freedom of compute-and-forward},''
  \emph{IEEE Trans. Inf. Theory}, vol.~58, no.~8, pp. 5214--5232, 2012.

\bibitem{Maddah-Ali2008}
M.~a. Maddah-Ali, A.~S. Motahari, and A.~K. Khandani, ``{Communication over
  MIMO X channels: Interference alignment, decomposition, and performance
  analysis},'' \emph{IEEE Trans. Inf. Theory}, vol.~54, no.~8, pp. 3457--3470,
  2008.

\bibitem{Sousa1990}
E.~S. Sousa and J.~a. Silvester, ``{Optimum transmission ranges in a
  direct-sequence spread-spectrum multihop packet radio network},'' \emph{IEEE
  J. Sel. Areas Commun.}, vol.~8, no.~5, pp. 762--771, 1990.

\bibitem{Jo2012}
H.~S. Jo, Y.~J. Sang, P.~Xia, and J.~G. Andrews, ``{Heterogeneous cellular
  networks with flexible cell association: A comprehensive downlink SINR
  analysis},'' \emph{IEEE Trans. Wirel. Commun.}, vol.~11, no.~10, pp.
  3484--3494, 2012.

\bibitem{Chun2015}
\BIBentryALTinterwordspacing
Y.~Chun, M.~Hasna, and A.~Ghrayeb, ``{Modeling Heterogeneous Cellular Networks
  Interference Using Poisson Cluster Processes},'' \emph{IEEE J. Sel. Areas
  Commun.}, 2015, to appear. [Online]. Available:
  \url{http://ieeexplore.ieee.org/lpdocs/epic03/wrapper.htm?arnumber=7110502}
\BIBentrySTDinterwordspacing

\end{thebibliography}

\begin{figure*}%
	\centering
		\begin{subfigure}{.65\columnwidth}
		\includegraphics[width=\textwidth, height=\textwidth]{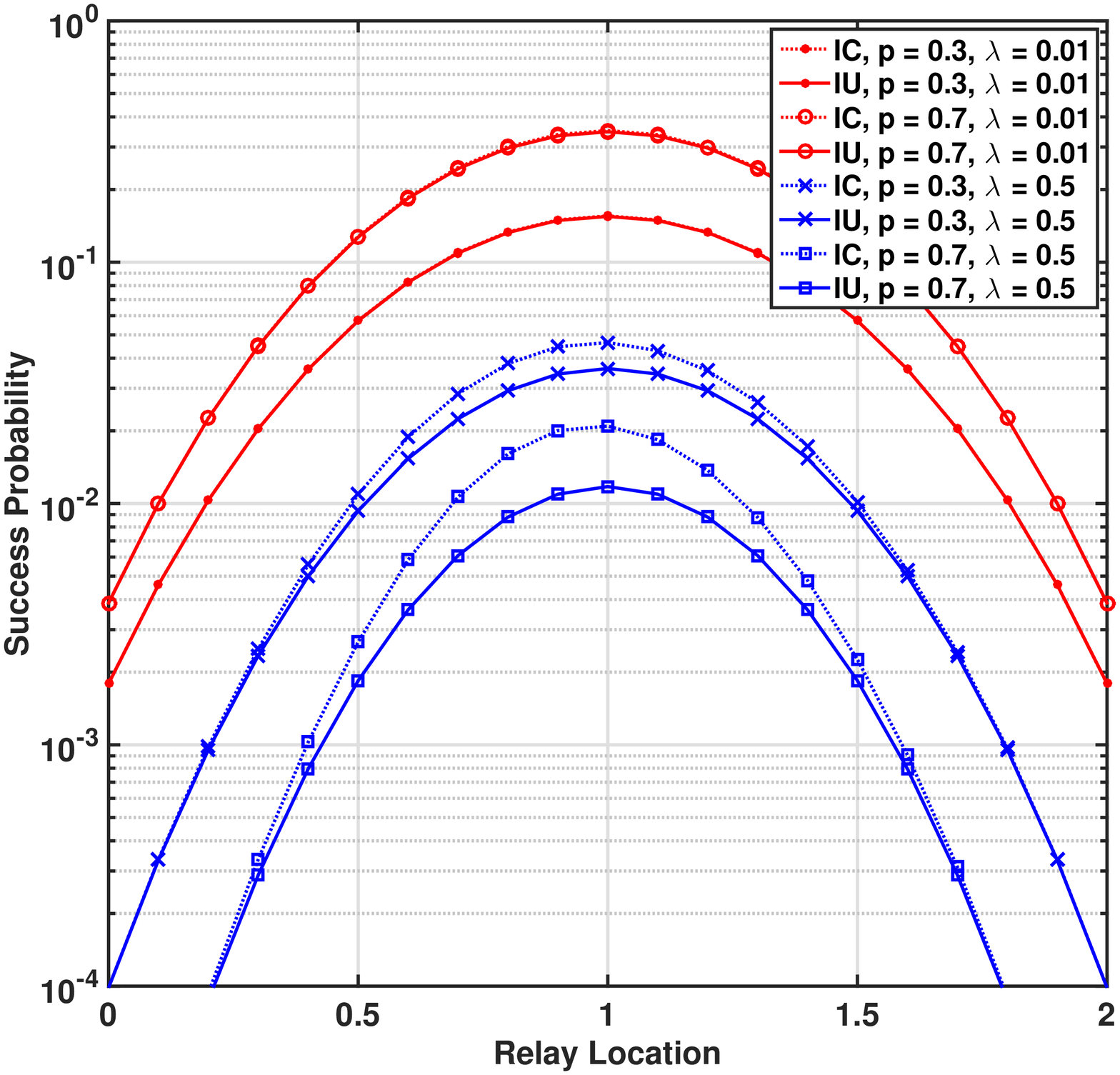}%
		\caption{}%
		\label{subfig-1a}%
		\end{subfigure}\hspace{0.1cm}%
			\begin{subfigure}{.65\columnwidth}
			\includegraphics[width=\textwidth, height=\textwidth]{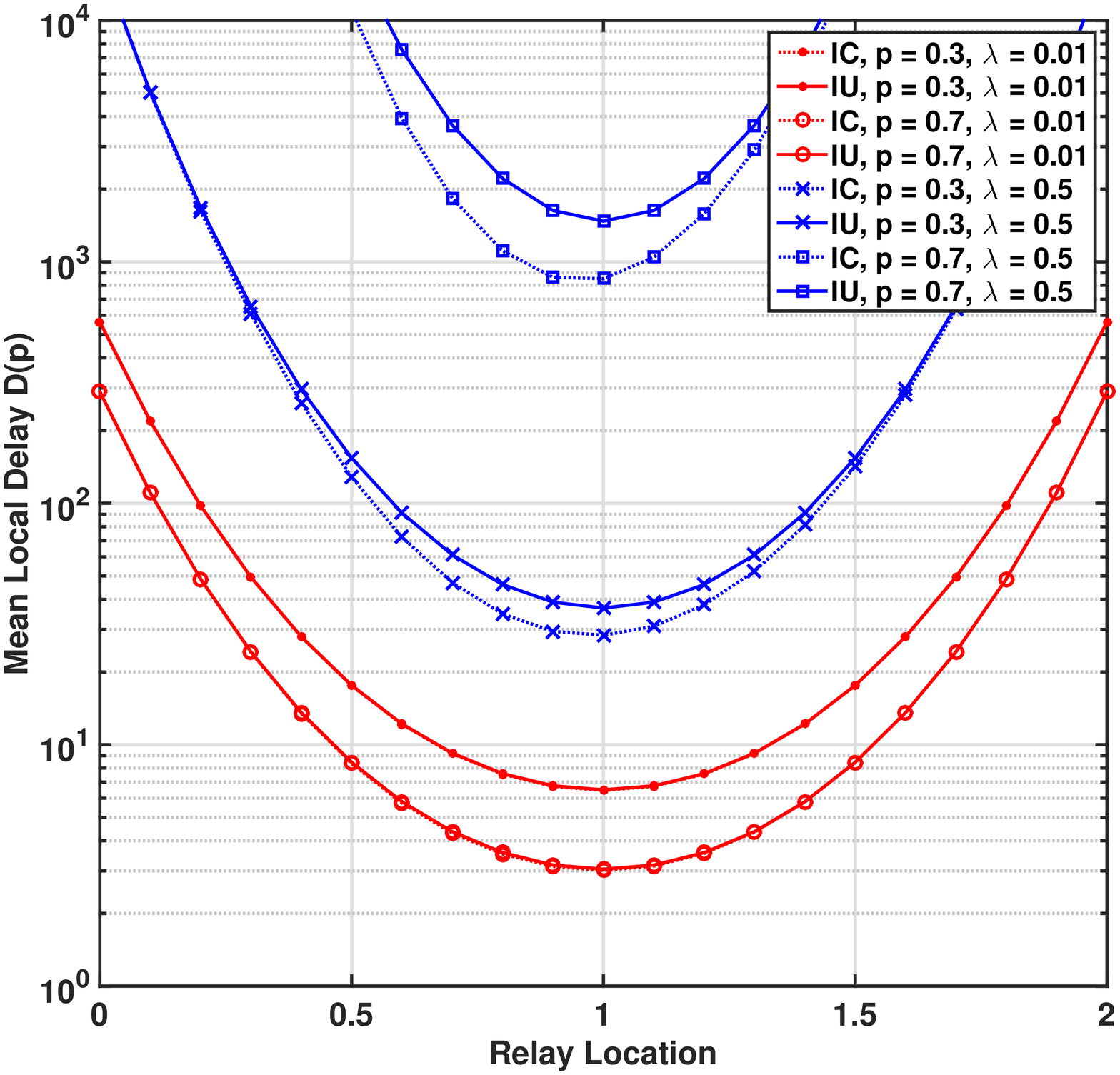}%
			\caption{}%
			\label{subfig-1b}%
			\end{subfigure}\hspace{0.1cm}%
		\begin{subfigure}{.65\columnwidth}
		\includegraphics[width=\textwidth, height=\textwidth]{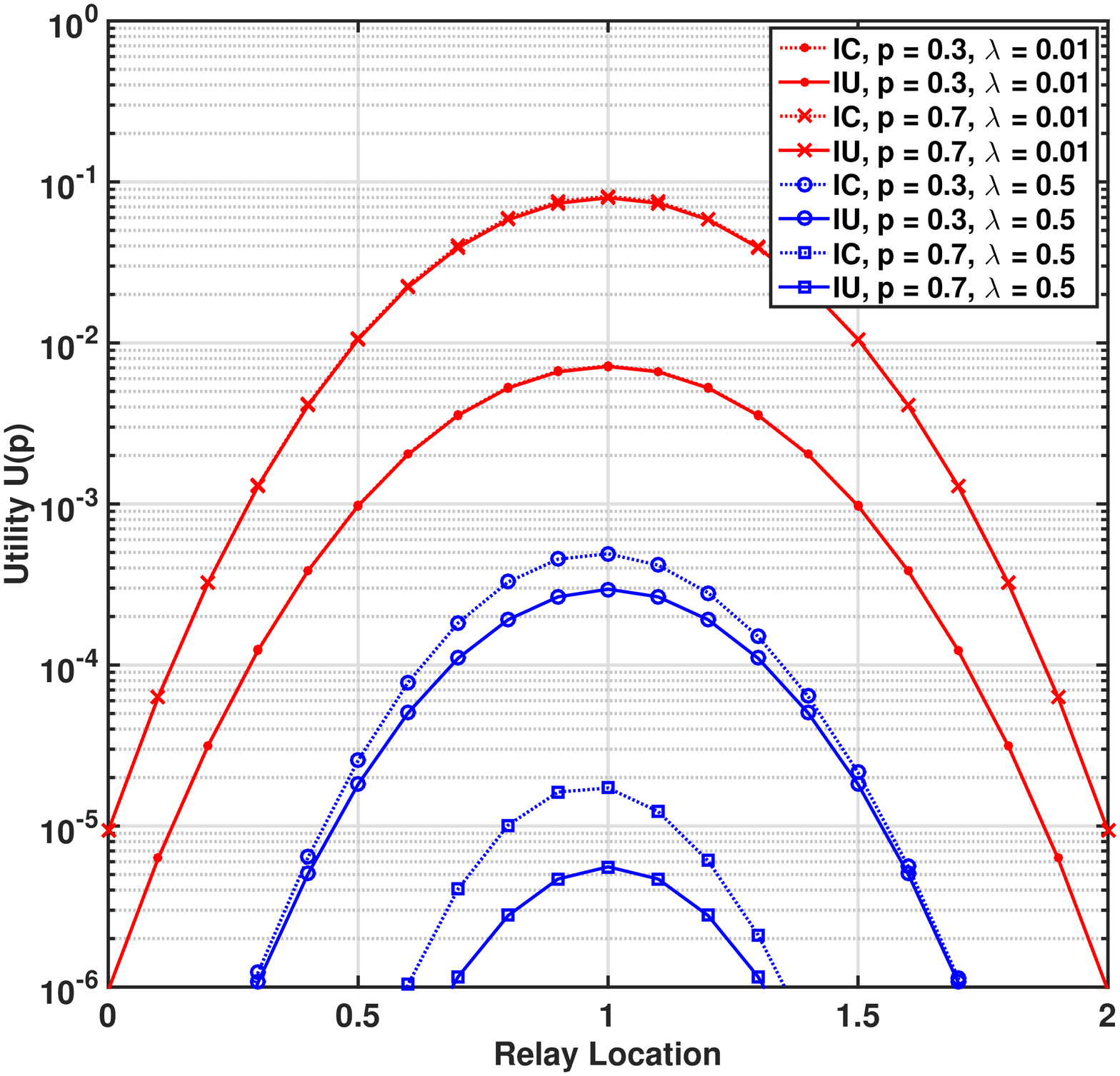}%
		\caption{}%
		\label{subfig-1c}%
		\end{subfigure}%
	\caption{(a) Success probability, (b) mean local delay, and (c) Utility versus relay location.}
	\label{fig1}
\end{figure*}

\begin{figure*}%
	\centering
		\begin{subfigure}{.65\columnwidth}
		\includegraphics[width=\textwidth, height=\textwidth]{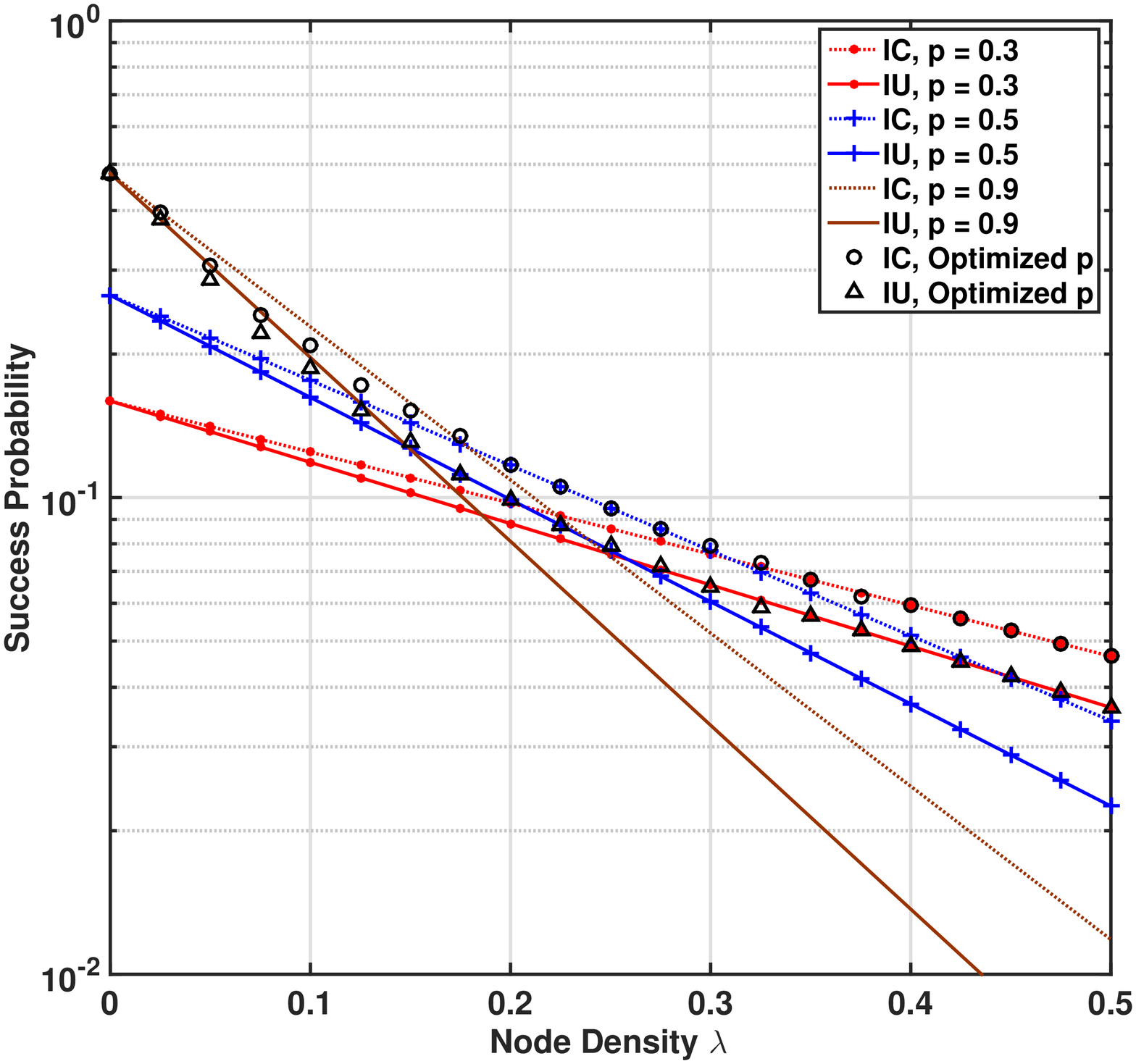}%
		\caption{}%
		\label{subfig-2a}%
		\end{subfigure}\hspace{0.1cm}%
			\begin{subfigure}{.65\columnwidth}
			\includegraphics[width=\textwidth, height=\textwidth]{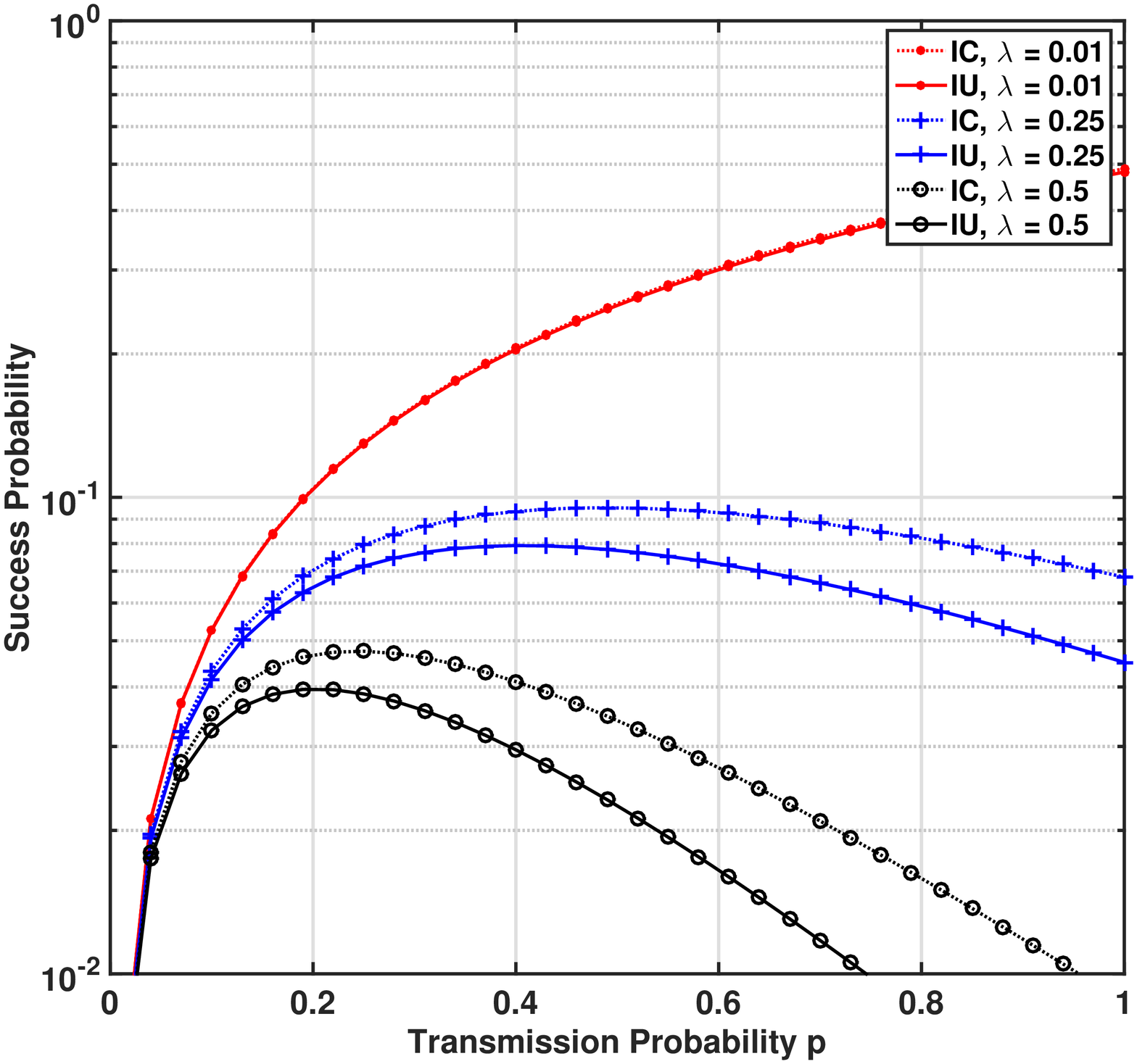}%
			\caption{}%
			\label{subfig-2b}%
			\end{subfigure}\hspace{0.1cm}%
		\begin{subfigure}{.65\columnwidth}
		\includegraphics[width=\textwidth, height=\textwidth]{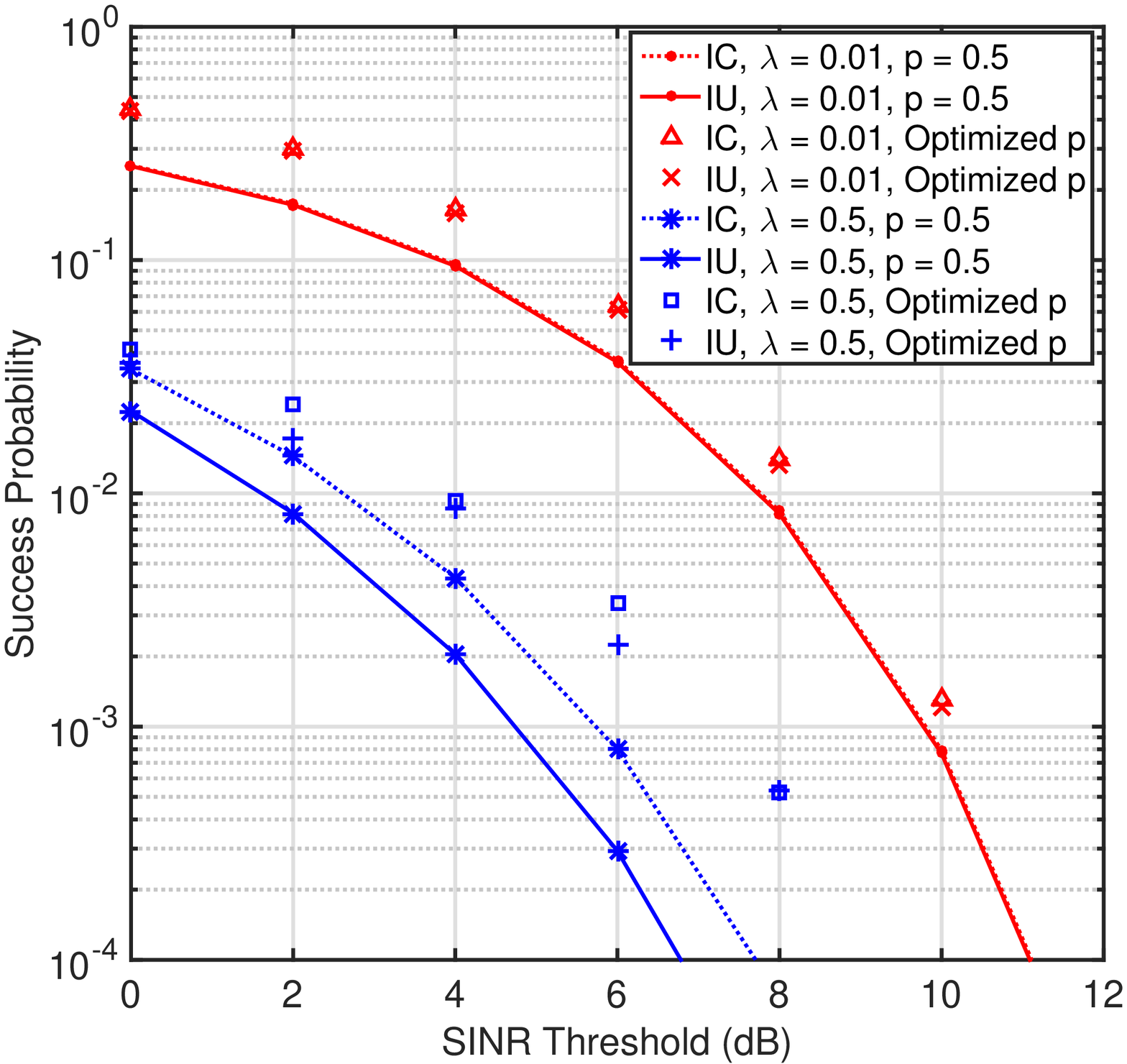}%
		\caption{}%
		\label{subfig-2c}%
		\end{subfigure}%
	\caption{Success probability versus (a) node density, (b) transmission probability, and (c) SINR threshold.}
	\label{fig2}
\end{figure*} 

\begin{figure*}%
	\centering
		\begin{subfigure}{.65\columnwidth}
		\includegraphics[width=\textwidth, height=\textwidth]{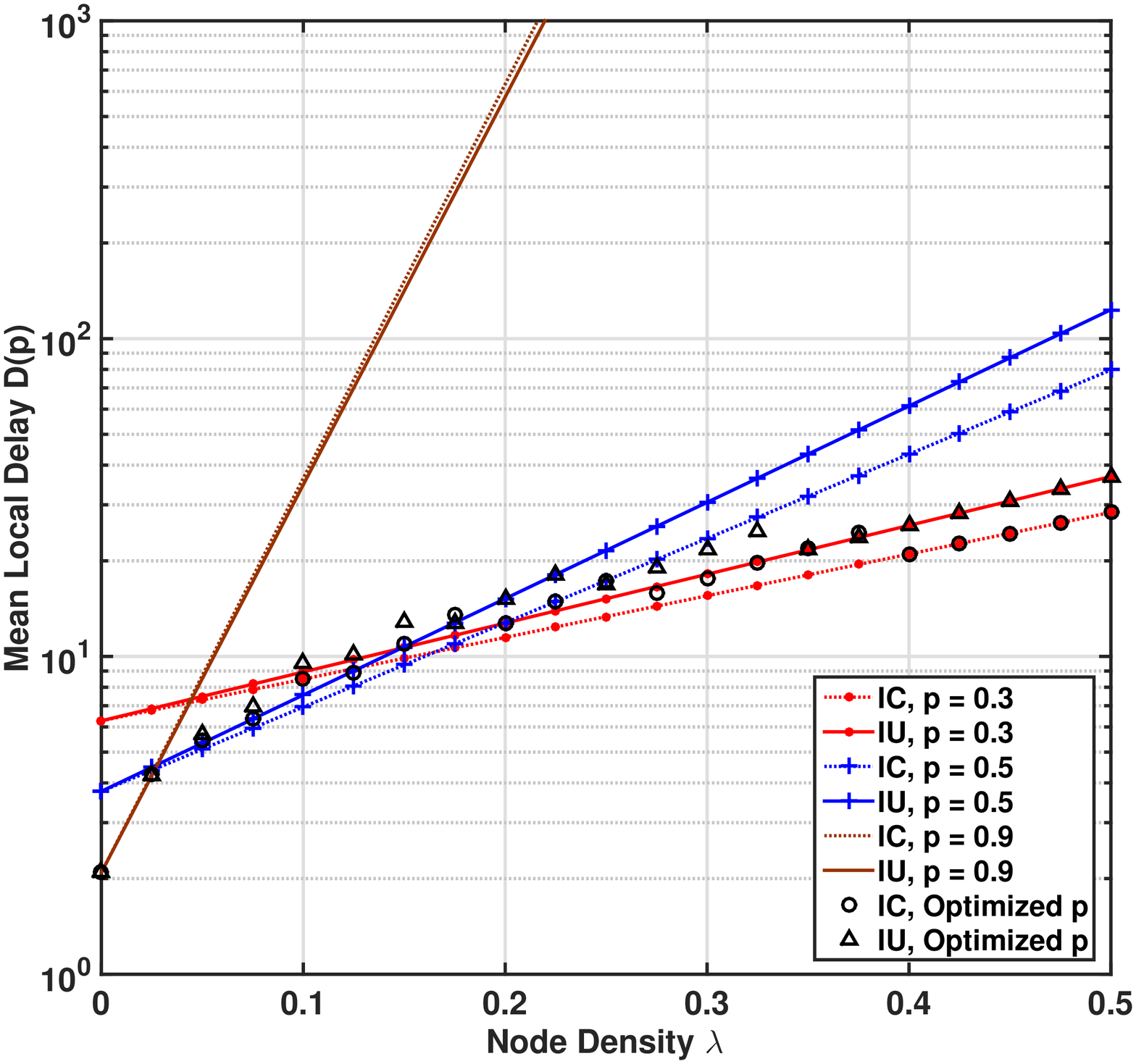}%
		\caption{}%
		\label{subfig-3a}%
		\end{subfigure}\hspace{0.1cm}%
			\begin{subfigure}{.65\columnwidth}
			\includegraphics[width=\textwidth, height=\textwidth]{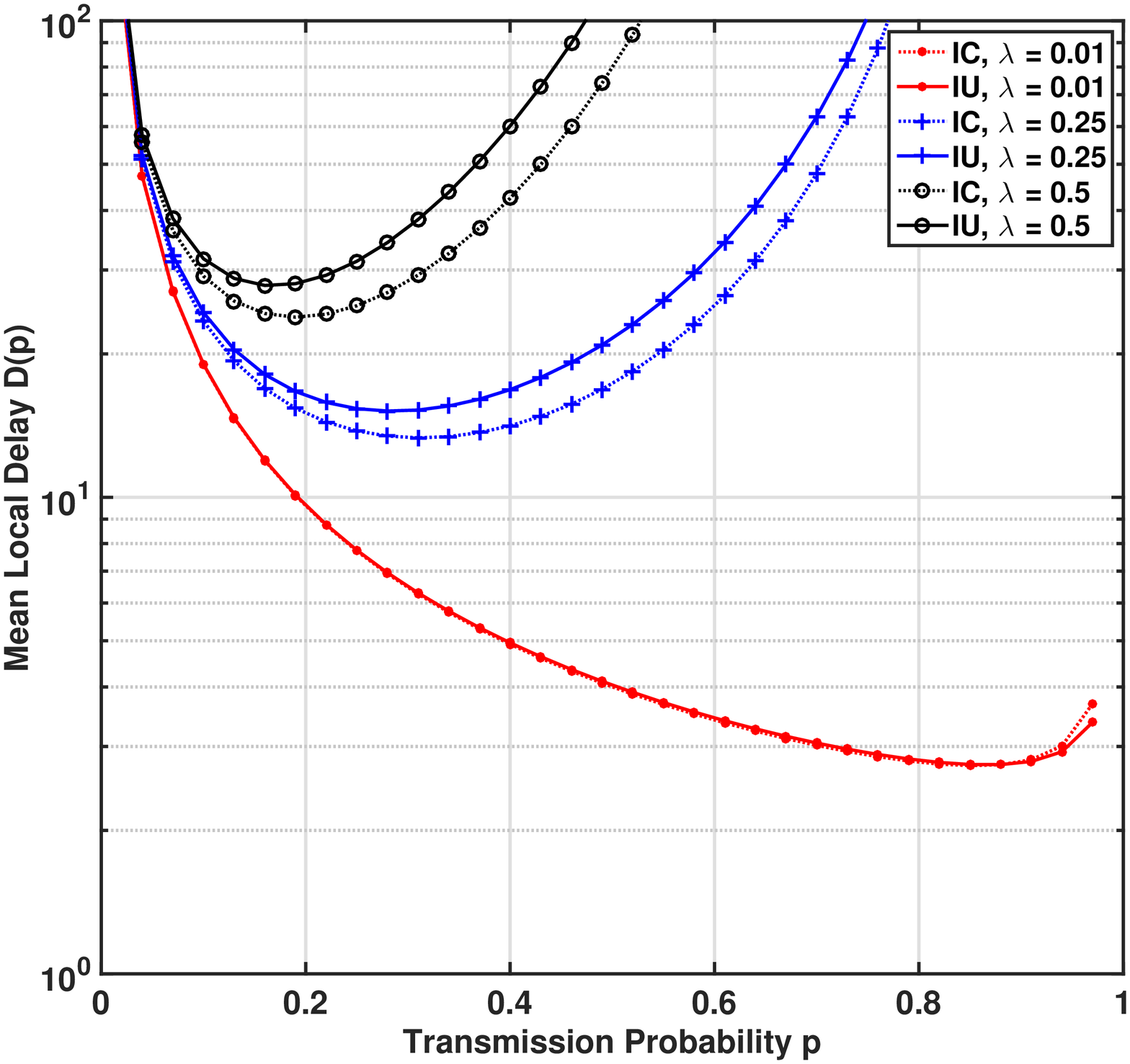}%
			\caption{}%
			\label{subfig-3b}%
			\end{subfigure}\hspace{0.1cm}%
		\begin{subfigure}{.65\columnwidth}
		\includegraphics[width=\textwidth, height=\textwidth]{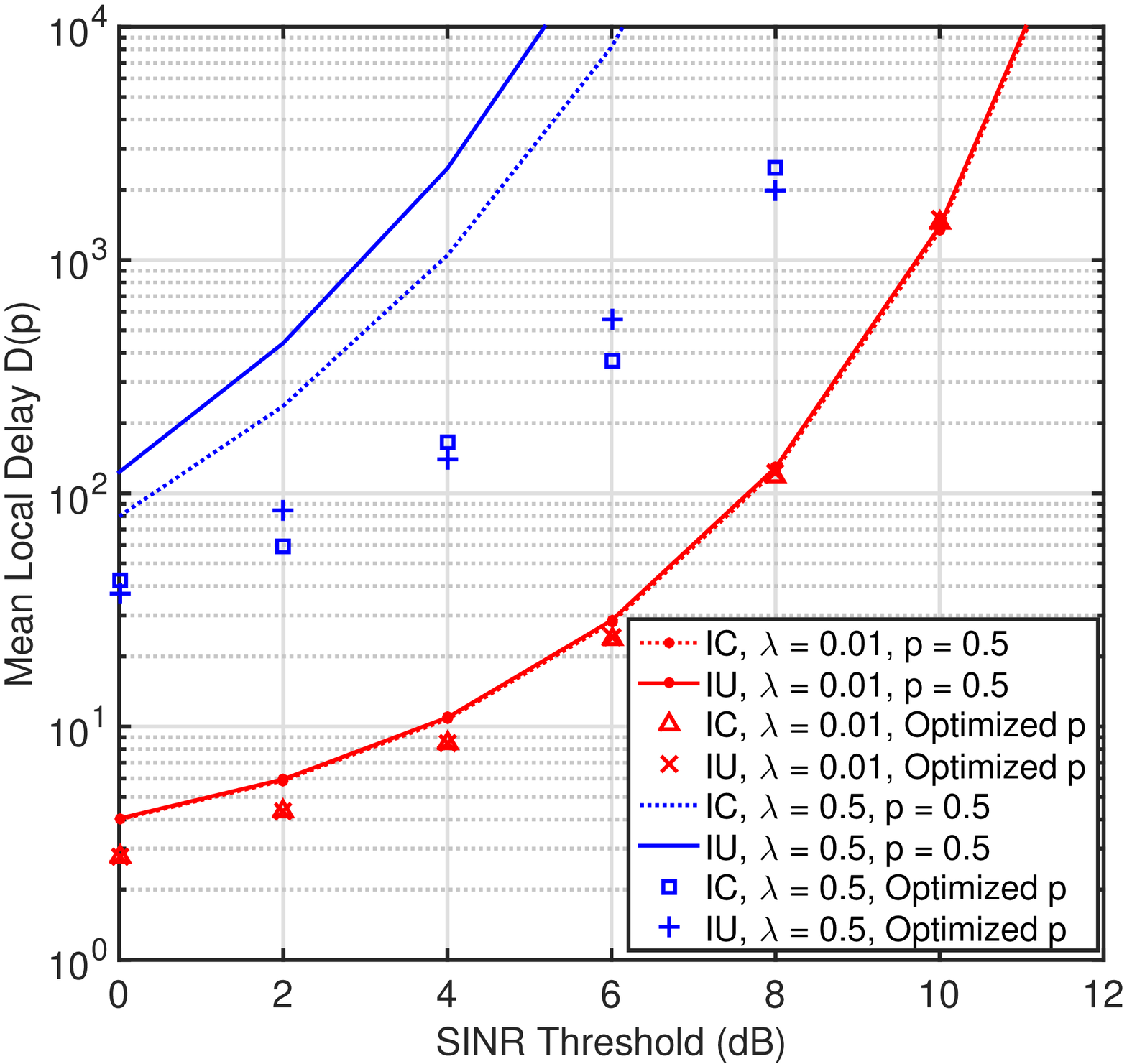}%
		\caption{}%
		\label{subfig-3c}%
		\end{subfigure}%
	\caption{Mean local delay versus (a) node density, (b) transmission probability, and (c) SINR threshold.} 
	\label{fig3}
\end{figure*}

\end{document}